\documentclass[11pt,a4paper]{article}
\usepackage{microtype}
\usepackage{amssymb, amsmath, amsthm}
\usepackage{fullpage}
\usepackage{enumerate}
\usepackage{graphicx}
\usepackage{subfig}
\usepackage{pifont}
\usepackage{hyperref}
\usepackage{wrapfig}

 \theoremstyle{plain}

\def\emphh#1{{\bf #1}}

\newtheorem{theorem}{Theorem}[section]
\newtheorem{lemma}[theorem]{Lemma}

\newtheorem{corollary}[theorem]{Corollary}

\theoremstyle{definition}
\newtheorem{remark}[theorem]{Remark}

\newif\iflong

\longtrue

\newif\ifdiag
\diagtrue

\begin{document}


\title{Embedding  graphs into embedded graphs}

\author{
{Radoslav Fulek\footnote{
The research leading to these results has received funding from the People Programme (Marie Curie Actions) of the European Union's Seventh Framework Programme (FP7/2007-2013) under REA grant agreement no [291734].}
}
}

\date{}

\maketitle

\begin{abstract}
A (possibly denerate) drawing of a graph $G$ in the plane is approximable by an embedding if it can be turned into an embedding by an arbitrarily small perturbation.
We show that testing, whether a straight-line drawing of a planar graph $G$ in the plane is approximable by an embedding, can be carried out in  polynomial time, if a desired embedding of $G$ belongs to a fixed isotopy class.
In other words, we show that c-planarity with embedded pipes is tractable for graphs with fixed embeddings.

To the best of our knowledge an analogous result was previously known essentially only when $G$ is a cycle.
 \end{abstract}

\section{Introduction}
 
In the theory of graph visualization  a drawing of a graph $G=(V,E)$ in the plane is usually assumed to be free of degeneracies, i.e., edge overlaps and edges passing through a vertex. However, in practice 
degenerate drawings often arise and need to be dealt with.

Recent papers~\cite{AAET16,ChEX15} address a certain aspect of this problem for simple polygons which can be thought of as straight-line (rectilinear)
embeddings of graph cycles.
  Chang et al.~\cite{ChEX15} gave an $O(n^2\log n)$-time algorithm to detect if a given polygon
with $n$ vertices can be turned into a simple (non self-intersecting) one by small perturbations of its vertices, or in other words if the polygon is \emphh{weakly simple}.
 We mention that there exists an earlier closely related definition of weakly simple polygons  by Toussaint~\cite{B93,T89}, however, as pointed out in~\cite{ChEX15} this notion is not well-defined for general polygons with ``spurs'',
 see~\cite{ChEX15} for an overview of attempts at combinatorial definitions
 of polygons not crossing itself.
 
An $O(n\log n)$ improvement on the running time of the algorithm by Chang et al.  was announced very recently by Akitaya et al.~\cite{AAET16}.
The combinatorial formulation of this problem corresponds to the setting of \emphh{c-planarity
with embedded pipes} introduced by Cortese et al.~\cite{CDPP09} well before the two aforementioned papers.
Therein only an $O(n^3)$-time algorithm for the problem was given.
Nevertheless, the algorithms in~\cite{AAET16,ChEX15} were built  upon the ideas from~\cite{CDPP09}.  Moreover, to the best of our knowledge the complexity status
of the c-planarity with embedded pipes is essentially known only for cycles.
Recently the problem was studied for general planar graphs  by Angelini and Da Lozzo~\cite{AL16}, but they gave only an FPT algorithm.
The introduction of this problem was motivated by a more general and well known problem of \emphh{c-planarity} by Feng et al.~\cite{FCEa95,FCEb95},  whose tractability status was open since 1995
 even in much more restricted cases than the one
that we consider. 
Biedl~\cite{B98} gave
 a polynomial-time algorithm for c-planarity with two clusters.
Beyond two clusters a polynomial time algorithm for c-planarity was obtained only in special cases,
e.g.,~\cite{BFPP08,GLS05,GJL+02,JJK+09,JKK+09}, and most recently in~\cite{BR14+,CBFK14+,Fb16}.

There is, however, another tightly related line of research on approximability or realizations of maps pioneered by Sieklucki~\cite{S69}, Minc~\cite{M97} and M.~Skopenkov~\cite{S03} that is completely independent from the aforementioned developments, and that is also a major source of inspiration for our work.
It can be easily seen that the result~\cite[Theorem 1.5]{S03} implies that c-planarity is tractable for flat instances with three clusters or cyclic clustered graphs~\cite[Section 6]{FKMP15} with a fixed isotopy class of a desired embedding. An algorithm with a better running time was given by the author in~\cite{Fb16}.

The aim of the present work is to show that c-planarity with embedded pipes is tractable for  planar graphs with a fixed isotopy class of embeddings, which extends results of~\cite{ADDF17,AL16,Fb16}. Our work also implies the tractability of deciding whether a drawing is approximable by an embedding in a fixed isotopy class, which extends  results of~\cite{AAET16,ChEX15}. This also answers in the affirmative a question posed in~\cite[Section 8.2]{ChEX15} if the isotopy class of an embedding of $G$ is fixed.  
 The combinatorial formulation of the  problem, \emphh{c-planarity with embedded pipes} follows, see Fig.~\ref{fig:problem}. We are given \\
 
\noindent {\bf (A)} A planar graph   $G=(V,E)$, whose vertex set is partitioned
into $k$ parts $V=V_{\nu_1}\uplus V_{\nu_2} \uplus \ldots \uplus
 V_{\nu_{k}}$ called \emphh{clusters} given by the isotopy class of an embedding of $G$ in the plane;\\
 {\bf (B)} a planar graph $H=(V(H), E(H))$ that is straight-line embedded in the plane,  where $V(H)=\{\nu_1,\ldots ,\nu_k\}$.

\begin{remark}
Since $H$ is straight-line embedded, $H$ does not contain multiple edges.
The assumption that $H$ is given by a straight-line embedding as opposed to a piecewise linear/polygonal embedding is not crucial, since every planar graph admits a straight-line embedding in the plane by F\'ary--Wagner theorem~\cite{F48,W36}, and its imposition is just a matter of convenience.
\end{remark}

 Let $\mathrm{dist}({\bf p}, {\bf q})$ denote the Euclidean distance between ${\bf p},{\bf q}\in \mathbb{R}^2$. Let 
 $\mathrm{dist}({\bf p}, S) = \min_{{\bf q}\in S} \mathrm{dist}({\bf p}, {\bf q})$, where  $S \subset \mathbb{R}^2$.
Let $N_{\varepsilon}(S)$ for $S \subset \mathbb{R}^2$ denote the $\varepsilon$-neighborhood
of $S$, i.e.,  $N_{\varepsilon}(S) = \{{\bf p }\in \mathbb{R}^2| \ \mathrm{dist}({\bf p},S)\leq \varepsilon\}$.
Let $\varepsilon,\varepsilon'>0$ be small values as described later.  
The \emphh{thickening} $\mathcal{H}$ of $H$ is the union
of $N_\varepsilon(\nu_i)$, for all $\nu_i\in V(H)$ and  $N_{ \varepsilon' }(\rho)$, for all $\rho \in E(H)$~\footnote{Throughout the paper we denote vertices and edges of $H$ by Greek letters.}. Let the \emphh{pipe} of $\rho\in E(H)$
be the closure of $N_{ \varepsilon' }(\rho)\setminus ( N_\varepsilon(\nu_i) \cup N_\varepsilon(\nu_j))$,
where $\rho=\nu_i\nu_j$. Let the \emphh{valve} of $\rho$ at $\nu_i$ be the curve obtained 
as the intersection of  $N_\varepsilon(\nu_i)$ and the pipe  of $\rho$.
We put $\varepsilon'<\varepsilon:=\varepsilon(H)$ so that the  valves are pairwise disjoint in $\mathcal{H}$ and
 $\varepsilon>0$ is smaller than $d/4$, where $d:=d(H)$ is the minimum distance between a vertex $\nu$ of $H$ and an edge $\rho$ of $H$ not incident to $\nu$ over all such edge-vertex pairs.

We want to decide if the given isotopy class of $G$ contains an embedding contained in $\mathcal{H}$, where the vertices in $V_{\nu_i}$, for every $\nu_i$, are drawn in the interior of $N_\varepsilon(\nu_i)$ and every edge crosses the boundary of $N_\varepsilon(\nu_i)$, for every $\nu_i\in V(H)$, at most once. Such an embedding of $G$ is \emphh{$H$-compatible}.
Let  $E_{\nu_i\nu_j}=\{uv\in E(G)| \ u\in V_{\nu_i}, v\in V_{\nu_j}\}$.
An $H$-compatible embedding of $G$ is encoded by $G,H$, and a set of total orders 
$(E_{\nu_i\nu_j}, <_{\omega})$,  for every $\nu_i\nu_j\in E(H)$ and a valve $\omega$ of  $\nu_i\nu_j$, where $(E_{\nu_i\nu_j}, <_{\omega})$ encodes the order of crossings of  $\omega$ with edges along $\omega$. The isotopy class of $G$ is encoded by a choice of the outer face, a set of rotations at its vertices and a containment relation of its connected components as described in Section~\ref{sec:pre}. Since we are interested only in  combinatorial aspects of the problem, $H$ is also given by   the isotopy class of its embedding.
Throughout the paper we assume that $G$ and $H$ are given as in {\bf (A)} and {\bf (B)}.

\begin{figure}[h]
\centering
\includegraphics[scale=1]{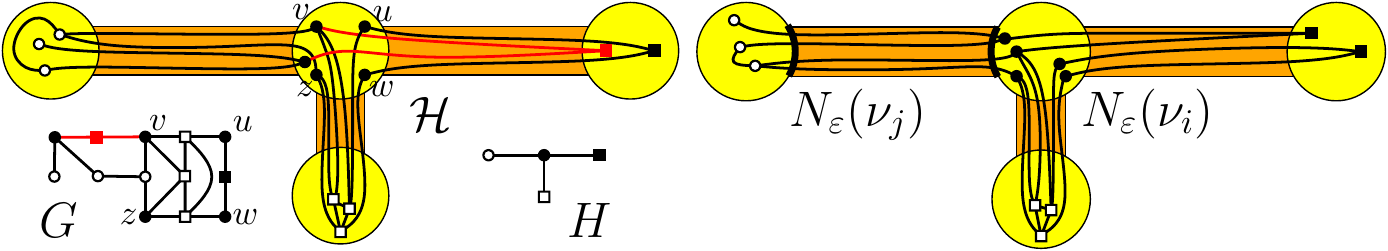}
\caption{Instance of c-planarity with embedded pipes. 
The partition of the vertex set of $G$ into clusters is encoded by the shape of vertices.
An $H$-compatible embedding of a subgraph of $G$ that cannot be extended to the whole $G$ (left). An $H$-compatible embedding of $G$ (right) inside $\mathcal{H}$.  The valves of $\rho=\nu_i\nu_j$ at $N_\varepsilon(\nu_i)$ and $N_\varepsilon(\nu_j)$ are highlighted by bold arcs.}
\label{fig:problem}
\end{figure}
\medskip

\begin{theorem}
\label{theorem:main}
There exists  $O(|V(G)|^2)$ algorithm that decides if the given isotopy class of $G$ contains an $H$-compatible embedding. An $H$-compatible embedding of $G$ can be also constructed in $O(|V(G)|^2)$ time if it exists.
In other words, c-planarity with embedded pipes is tractable, when an  isotopy class of  a desired embedding of $G$ is fixed.
\end{theorem}

As a corollary of our result we obtain that we can test in  polynomial time if a piecewise linear drawing of a graph in the plane is approximable by an embedding and construct such an embedding if it exists. 
We defer the definition of the approximability by an embedding to Section~\ref{sec:approx},
where also the proof of the corollary can be found.
As previously discussed this extends results in~\cite{AAET16,ChEX15} and also~\cite{S03}. 

\begin{corollary}
\label{cor:main}
There exists  an $O(n^4)$ time algorithm that decides if
 a piecewise linear drawing of a graph in the plane is approximable by an embedding, and constructs such an embedding if it exists, where $n$ is the size of the representation of the drawing.
\end{corollary}

{\bf Extensions of our results.}
By~\cite[Theorem 3.1]{M06}, our result holds also in the setting of  rectilinear, i.e., straight-line, drawings of graphs. To extend it further in this setting by allowing ``forks'' (see Section~\ref{sec:approx}) seems to be just a little bit technical.

%

In a recent manuscript~\cite{FK17+}, we verified a conjecture of M.~Skopenkov~\cite[Conjecture 1.6]{S03} implying that that our problem is tractable, when we lift the restriction on the isotopy class $G$.  This does not imply that the problem with the restriction  on the isotopy class $G$ is tractable except when $G$ is connected. The running time of the algorithm, that is implied by this work,
is $O(|V|^{2\omega})$, where $O(n^\omega)$ is the running time of the fastest algorithm for
multiplying a pair of $n$ by $n$ matrices.
 Since $\omega> 2$ due to the matrix size, this is much worse that the running time claimed by Theorem~\ref{theorem:main}. Furthermore, the algorithm is not constructive.

As noted by Chang et al.~\cite{ChEX15}, the technique of Cortese et al.~\cite{CDPP09} extends directly  from the plane to any closed two-dimensional surface.  The same holds for our method, but since considering
general two-dimensional surfaces does not bring anything substantially new to our treatment of the problem, for the sake of simplicity  we consider only the planar case. \\

 {\bf Strategy of the proof of Theorem~\ref{theorem:main}.}
Formally, the \emphh{input} of our algorithm is a triple $(G,H,\gamma)$,
where the partition of the vertex set of $G$ corresponds to the map $\gamma$ from the set of vertices  of $G$ to the set of vertices of $H$. Hence,  for $v\in V_{\nu}$, where $\nu\in V(H)$, we have  $\gamma(v)=\nu$. The input $(G,H,\gamma)$ is \emphh{positive} if there exists an $H$-compatible
embedding of $G$ in the given isotopy class of $G$, and \emphh{negative} otherwise.

The most problems in constructing a polynomial time algorithm  for our problem
are caused by so called ``spurs'' such as the red vertex in Fig.~\ref{fig:problem} (left), i.e., connected components in subgraphs of $G$ induced by clusters, whose all adjacent
vertices belong to the same cluster.
Due to the presence of spurs it is hard to see that our problem is tractable even in the case,
when $G$ is a path.

The centerpiece of our method is an extension of the definition of the derivative of maps of intervals/loops (corresponding to the case, when $G$ is a path/cycle, in our terminology) in the plane introduced by Minc~\cite{M97} to arbitrary graphs. 
  We adapt this notion to the setting of c-planarity with embedded pipes.
The derivative is an operator that takes $(G,H,\gamma)$, and either detects that 
there exists no $H$-compatible embedding of $G$ in the given isotopy class of $G$, or outputs $(G',H',\gamma')$, that  is also a valid input for our algorithm, such that $(G,H,\gamma)$ is positive
if and only if   $(G',H',\gamma')$ is positive.
Intuitively, $H'$ is reminiscent of the line graph of $H$ and the subgraphs of $G$, that are mapped by $\gamma$ to the edges of $H$, are turned into subgraphs of $G'$ mapped by $\gamma'$ into vertices of $H'$. This results in a shortening of problematic spurs, and zooming into the structure of the map $\gamma$.
We show that by iterating the derivative $|E(G)|$ times we either detect 
that there exists no $H$-compatible embedding of $G$ in the given isotopy class of $G$, or
we arrive at an input without problematic spurs.
Since it is fairly easy to solve the problem for the latter inputs; the derivative at every iteration
can be computed in linear time in $|V(G)|$; and by derivating the size of the input is increased only by a little,  the tractability follows.

The operation of node expansion and base contraction introduced in~\cite{CDPP09} resemble the derivative. The main difference is that these two operations affect only a single cluster or a pair of clusters in $(G,H,\gamma)$, and therefore they are local, whereas  the derivative changes the whole input.
We are very positive that our method is applicable to other graph drawing problems related to c-planarity whose tractability is open.
This is documented by our recent manuscript~\cite{FK17+} in which 
a similar technique was applied.

The derivative is applied to an input $(G,H,\gamma)$, in which every cluster $V_{\nu_i}$ induces
in $G$ an independent set. Such an input is in the \emphh{normal form}.
The detailed description of the algorithm proving Theorem~\ref{theorem:main} is in Section~\ref{sec:algorithm}.
We show in Section~\ref{sec:instances} that an input can be assumed to be in the normal form.
The definition of the  derivative is given in Section~\ref{sec:derivative},
and sufficiently simplified inputs are dealt with in Section~\ref{sec:loc_inj}. \\




\section{Preliminaries}

\label{sec:pre}

Throughout the paper we tacitly use  Jordan-Sch\"onflies theorem.

Let $G=(V,E)$ denote a   planar graph possibly with multiple edges and loops.
 For $V'\subseteq V$ we denote by $G[V']$ the sub-graph of $G$ induced by $V'$. 
A \emphh{star} ${St(v)}$ of a vertex $v$ in a graph $G$ is the subgraph of $G$ 
consisting of all the edges incident to $v$.
Throughout the paper we use standard graph theoretical notions such as path, cycle, walk, vertex degree $deg(v)$ etc., see~\cite{D10}.

A \emphh{drawing}  $\mathcal{D}(G)$ is a representation of $G$ in the plane, where every vertex
 in $V$ is represented by a  point and every
edge $e=uv$ in $E$ is represented by a simple piecewise linear curve joining the points that represent $u$ and $v$.
 Thus, a drawing can be thought of as a map from $G$ understood as a topological space into the plane.
In a drawing, we additionally require every pair of distinct curves representing edges to meet only 
in finitely many points each of which is a proper crossing or a common endpoint. In a \emphh{degenerate} drawing, we allow a pair of distinct vertices to be represented by the same point and a pair of edges to be represented by the same  curve. 
A drawing in which every vertex is represented by a unique point and 
every edge by a unique curve is \emphh{non-degenerate}. 
In a non-degenerate drawing, multiple edges are mapped to distinct arcs meeting at their endpoints.
In the paper we consider non-degenerate drawings,
except in Section~\ref{sec:approx}.
An  edge crossing-free non-degenerate drawing is an \emphh{embedding}. A graph given by an embedding in the plane
is a  \emphh{plane graph}.
If it leads to no confusion, we do not distinguish between
a vertex or an edge and its representation in the drawing and we use the words ``vertex'' and ``edge'' in both
 contexts.
 
 The following lemma is well known.
 \begin{lemma}
 \label{lemma:linear}
 Let $G$ be a plane graph with $n$ vertices 
 such that   $G$ does not contain a pair of multiple edges joining the same pair of vertices 
that form a face of size two, i.e., a lens,
except for the outer face. 
The graph $G$ has $O(n)$ edges.
 \end{lemma}

The \emphh{rotation} at a vertex in an embedding of $G$ is the counterclockwise cyclic order of the end pieces of its incident edges.
The rotation at  a vertex is stored as a doubly linked list of edges. Furthermore, we assume that for every edge of $G$ we store a pointer to its preceding and succeeding edge in the rotation at both of its end vertices.
The \emphh{interior} and \emphh{exterior} of a cycle in an embedded graph is the bounded and unbounded, respectively, connected component
of its complement in the plane. 
Similarly, the \emphh{interior} of an inner face and outer face in an embedded
connected graph is the bounded and unbounded, respectively, connected component
of the complement of its facial walk in the plane bounded by the walk.
An embedding of a connected graph $G$ is up to an isotopy described by the rotations at its vertices and the choice of its outer (unbounded) face.
If $G$ is not connected the isotopy class
of its embedding is described by isotopy classes
of its connected components $G_1,\ldots, G_l$
and the containment relation $G_i \subset f$, for every $G_i$, where
$f$ is a face of $G_j$, $j\not= i$, such that $G_i$ is embedded in the interior of $f$.


\section{Approximation of maps by embeddings}
\label{sec:approx}

The aim of this section is to derive Corollary~\ref{cor:main} from Theorem~\ref{theorem:main}.
By treating a graph $G$ as a one-dimensional topological space, a drawing $\mathcal{D}$
of $G$ is understood as a continuous map $\mathcal{D}$ mapping every $x\in G$ to $\mathbb{R}^2$.
Such a drawing is given by the finite set of pairs of real values representing the end points
of line segments of  polylines corresponding in the drawing to edges of $G$.

Let $G=(V,E)$ denote a planar graph. Let $\mathcal{D}$ be a (possibly degenerate) drawing corresponding of $G$.
Note that we do not 
allow an edge to pass through a vertex by the definition of the drawing, or in other words, we do not allow a drawing to contain \emphh{forks}~\cite{ChEX15}. 
 The previous restriction is not crucial, since we can subdivide edges at ``fork'' vertices while still having an input of a quadratic size in the size of the original input.
 This yields the claimed running time.
An \emphh{$\epsilon$-approximation} of a drawing  $\mathcal{D}$ of a graph $G$ is a drawing $\mathcal{D}'$ of $G$ such that $\mathrm{dist}(\mathcal{D}(x),\mathcal{D}'(x))<\epsilon$,
for all $x\in G$.
 A drawing $\mathcal{D}$ is \emphh{approximable by an embedding} if  there exists 
$\varepsilon(\mathcal{D})>0$ such that
for every  $\varepsilon$, $0<\varepsilon<\varepsilon(\mathcal{D})$, there exists an $\varepsilon$-approximation $\mathcal{D}$ that is an embedding. 
It is clear that a  drawing with edge crossings is not 
approximable by an embedding, and thus, in the sequel we consider only drawings without edge crossings.

Given a drawing $\mathcal{D}$ of a graph $G$ in the plane, in order to decide if $\mathcal{D}$ is approximable by an embedding in a fixed isotopy class of $G$, we construct an  input $(G_0,H,\gamma)$ for c-planarity with embedded pipes. The graph $H$ is the embedded graph given by the image of $\mathcal{D}$, and $G_0$ is obtained from $G$ by subdividing every edge $e$, that is not a drawn as a straight-line segment by $\mathcal{D}$, so that $\mathcal{D}$ is turned into a straight-line drawing of $G_0$.  Then every end point of a line segment  representing an   edge of $G_0$ is turned into a vertex of $H$. We put $\gamma(v):=\mathcal{D}(v)$ for $v\in V(G_0)$.

The input $(G_0,H,\gamma)$ is positive if and only if  $\mathcal{D}$ is approximable by an embedding. 
The ``only if'' direction is easy.
If $(G_0,H,\gamma)$ is positive, then   we can choose $\varepsilon(\mathcal{D}):=\varepsilon$, where $\varepsilon$ is as in the definition of the thickening of $H$, witnessing that  $\mathcal{D}$ is approximable by an embedding.
Let $\varepsilon'$ be as in the definition of the thickening of $H$.
To prove the ``if'' direction, it is enough to show that an $\varepsilon'$-approximation of $\mathcal{D}$, that is an embedding, can be chosen such that for every   
$\nu_i\nu_j\in E(H)$,  $v_i\in V_{\nu_i}$ and $v_j\in V_{\nu_j}$ we have $|\omega \cap \mathcal{D}(v_iv_j)|\leq 1$, where $\omega$ is a valve of $\nu_i\nu_j\in E(H)$.
This can be achieved 
by an appropriate local deformation of the $\varepsilon'$-approximation as we show next.

\begin{figure}
\centering
\subfloat[]{
\includegraphics[scale=1]{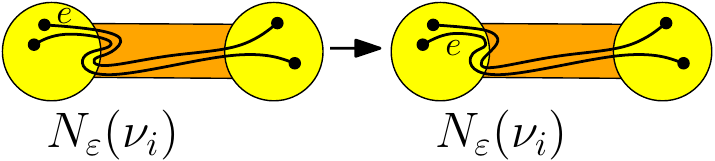}\label{fig:lensremoval}}
\subfloat[]{
\includegraphics[scale=1]{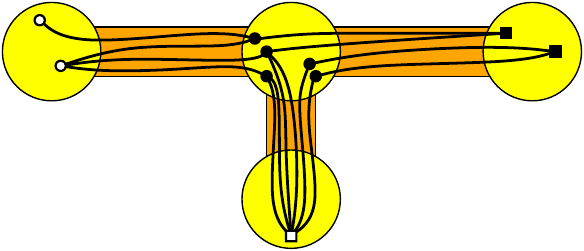}\label{fig:normal}}
\caption{(a) Deforming an approximation so that every valve is crossed by an edge at most once.
(b) The input $(G^N,H^N,\gamma^N)$ in the normal form constructed from $(G,H,\gamma)$ in Fig.~\ref{fig:problem}.}
\end{figure}

Suppose that a valve $\omega$ at, let's say $\nu_i$, crosses an edge at least two times in  the $\varepsilon'$-approximation. We consider a pair of consecutive crossings with $\omega$ along an edge $e$ such that the piece of $e$ between the  crossings in the pair is contained in the pipe. We choose the pair so that the distance between the crossings is minimal, and eliminate the crossings as illustrated in Fig.~\ref{fig:lensremoval}.
By repeating this procedure, we eventually obtain an $H$-compatible
embedding of $G_0$.

\section{Proof of Theorem~\ref{theorem:main}}
\label{sec:algorithm}

Let $(G,H,\gamma)$ be the input of our algorithm,
where the partition of the vertex set of $G$ corresponds to the map $\gamma$ of the vertices  of $G$ by vertices of $H$. Recall that  for $v\in V_{\nu}$ we have  $\gamma(v)=\nu$,
where $\nu\in V(H)$. 
We also naturally extend $\gamma$ to edges: $\gamma(v_iv_j)=\rho=\nu_i\nu_j$, for $v_i\in V_{\nu_i}$ and $v_j\in V_{\nu_j}$, and to subgraphs $G_0$ of $G$: $\gamma(G_0)=H_0=(V(H_0),E(H_0))$ such that $V(H_0)=\{\nu\in V(H)| \gamma(v)=\nu, \ v\in V(G_0 )\}$ and $E(H_0)=\{\rho\in E(H)|  \ \gamma(e)=\rho, e\in E(G_0) \}$.

A vertex $\nu\in V(H)$ of degree two  is \emphh{redundant} if $V_{\nu}\subseteq V(G)$ is an independent set  consisting of vertices of degree two such that for every $v\in V_{\nu}$ we
have $\gamma(vu)\not=\gamma(vw)$, where $u$ and $w$ are the two neighbors of $v$.
We assume that every edge of $H$ is used by at least one edge of $G$, i.e., for every $\rho\in E(H)$ there exists $e\in E(G)$ such that $\gamma(e)=\rho$.

\subsection{The normal form} 
\label{sec:instances}

Similarly as in~\cite{Fb16}, the input $(G,H,\gamma)$ is in the \emphh{normal form} if
\begin{enumerate}[(1)]
\item
\label{it:ind}
every cluster $V_{\nu}\subseteq V(G)$, for $\nu\in V(H)$, is an independent set without isolated vertices; and
\item
\label{it:supp}
$H$ does not contain a pair of redundant vertices joined by an edge.
\end{enumerate}

We remark that~(\ref{it:supp}) is required only due to the running time analysis. Then we do not forbid redundant vertices  completely, since we do not allow $H$ to contain multiple edges.
 In what follows we show 
how to either detect that no $H$-compatible embedding in the given isotopy class of $G$
exists just by considering the subgraph of $G$ induced by a single cluster $V_{\nu}$,
or construct an input $(G^N,H^N,\gamma^N)$, see Fig.~\ref{fig:normal}, in the normal form, which is positive if and only if
the input $(G,H,\gamma)$ is positive.
Clearly,~(\ref{it:supp}) can be assumed without loss of generality. Before establishing the other condition we introduce a couple of definitions. 

A \emphh{contraction} of an  edge $e=uv$ in an embedding of a graph is an operation that turns
$e$ into a vertex
by moving $v$ along $e$ towards $u$ while dragging all the other edges incident to $v$ along $e$.
By a contraction we can introduce multiple edges or loops at the vertices.
We will also  use the following operation which can be thought of as the inverse operation of the edge contraction in an embedding of a graph.
Note that a contraction can be carried out in $O(1)$ time, since it amounts to merging a pair of doubly linked lists, and redirecting at most four pointers. The same applies to the following operation.
A \emphh{vertex split}, see Fig.~\ref{fig:splits}, in an embedding of a graph $G$ is an operation that replaces a vertex $v$ by two vertices $u$ and $w$
joined by a crossing free edge so that the neighbors of $v$ are partitioned into two parts
according to whether they are joined with $u$ or $w$ in the resulting drawing. The rotations at $u$ and $w$ are inherited from the rotation at $v$.
When applied to $G$, the operations are meant to return a graph given by an isotopy class of its embedding; the same applies to 
vertex multisplit defined later.

In order to satisfy~(\ref{it:ind}), by a series of successive edge contractions we contract each connected component of $G[V_{\nu_i}]$, for $\nu_i\in V(H)$, to a vertex.
Since rotations are stored as doubly linked lists, contracting all such connected components can be carried out in linear time.
We delete any created loop and isolated vertices.
If a loop at a vertex from $V_{\nu_i}$ contains a vertex from a different cluster $V_{\nu_j}$, $\nu_j\not=\nu_i$, in its interior we know that the 
input is negative, since for every $\nu_j$ all the vertices in $V_{\nu_j}$ must be contained in the outer face of $G[V_{\nu_i}]$ if the input is positive. All this can be easily checked in  linear time in $|V(G)|$ by the the breadth-first
or depth-first search algorithm.
If a loop at a vertex from $V_{\nu_i}$ does not contain  a vertex from a different cluster, such a contraction preserves the existence of an $H$-compatible embedding in the given isotopy class of $G$. Indeed, isolated vertices and
deleted empty loops can be introduced in an $H$-compatible embedding of the resulting  graph, and contracted edges recovered via vertex splits.
Let $(G^N,H^N,\gamma^N)$ denote the resulting input in the normal form. We  proved the following.

\begin{lemma}
\label{lemma:normal}
If a loop at a vertex of $G$ obtained during the previously described procedure  contains a vertex of $G$ in its interior the input $(G,H,\gamma)$ is negative.
Otherwise, the input $(G,H,\gamma$) is positive if and only if $(G^N,H^N,\gamma^N)$ is positive.
\end{lemma}


%

\subsection{Derivative}
\label{sec:derivative}

We present the operation of the derivative that simplifies the input, and whose
iterating results in an input that is easy to deal with.
 Such inputs are treated in Section~\ref{sec:loc_inj}.
Before we describe the derivative we give a couple of definitions.

A \emphh{vertex multisplit}, see Fig.~\ref{fig:splits}, in an embedding of a graph $G$ is an operation that replaces a vertex $v$ with a crossing free star $(\{v,v_1,\ldots, v_l\}, \{vv_1,\ldots, vv_l \})$   so that
 the resulting underlying graph has vertex set $V(G)\cup \{v_1,\ldots, v_l\}$ and edge set $(E(G)\setminus \{vu_1,\ldots ,vu_{deg(v)}\})\cup \{v_{i_j}u_j| \ j=1,\ldots, deg(v)\}\cup
 \{vv_1,\ldots, vv_l \}$, where $u_1,\ldots, u_{deg(v)}$
  are neighbors of $v$ in $G$ and  $1\le i_j \le l$, for all $j$.
The rotations at $v_1,\ldots,v_l$ are inherited from the
rotation at $v$ so that by contracting all the edges of $St(v)$ in the resulting 
graph we obtain the original embedding of $G$.
Note that a vertex multisplit can be carried out in $O(deg(v))$ time.

The rotation of $\nu\in V(H)$ is \emphh{consistent} with the rotation of $v\in V_{\nu}$ if
the rotation given by ($\gamma(vv_1),\ldots, \gamma(vv_{deg(v)}))$, where $(vv_1,\ldots, vv_{deg(v)})$
is the rotation at $v$ in an embedding of $G$ in the given isotopy class, is the rotation at $\nu\in V(H)$ in the embedding of $H$.
Let the \emphh{potential} $p(G,H,\gamma)=|E(G)|-|E(H)|$.
Obviously, $p(G,H,\gamma)\ge 0$ and $p(G,H,\gamma)= 0$,
if $G$ is isomorphic to $H$ via $\gamma$,
and if $G$ is connected the opposite implication also holds.
The input in the normal form $(G,H,\gamma)$ is \emphh{locally injective} if

\begin{enumerate}[(i)]
 \item \label{it:li}
 the restriction of $\gamma$ to $V(St(v))$ is injective, for all $v\in V(G)$; and
\item \label{it:v32} every vertex $v$ of degree one
in $G$ is incident to an edge $e$ such that $\gamma(e)=\gamma(f)$ implies $e=f$ for all $f\in E(G)$. 
\end{enumerate}

Given an input $(G,H,\gamma)$, the vertex $v\in V(G)$ is \emphh{fixed}  if the condition of property~(\ref{it:li}) holds 
for $v$, and $v$ is alone in its cluster, i.e.,
$\gamma(u)=\gamma(v)$ implies $u=v$.
If $v$ is fixed then we call $\gamma(v)=\nu_i \in V(H)$ also \emphh{fixed}.

Given an input $(G,H,\gamma)$ in the normal form that is not locally injective, we either detect that there does not exist an $H$-compatible embedding of $G$ in the given isotopy class, or we construct the input $(G',H',\gamma')$ having a smaller potential after being brought to the normal form, such that
$(G',H',\gamma')$ is positive if and only if  $(G,H,\gamma)$
is positive. 
The input $(G',H',\gamma')$ is obtained as follows, see Fig.~\ref{fig:derivative1}.

\begin{figure}
\centering
\subfloat[]{\includegraphics[scale=1]{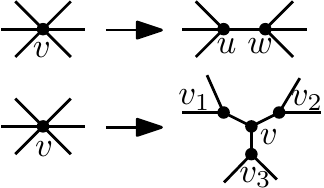}
\label{fig:splits}}
\hspace{10pt}
\subfloat[]{\includegraphics[scale=1]{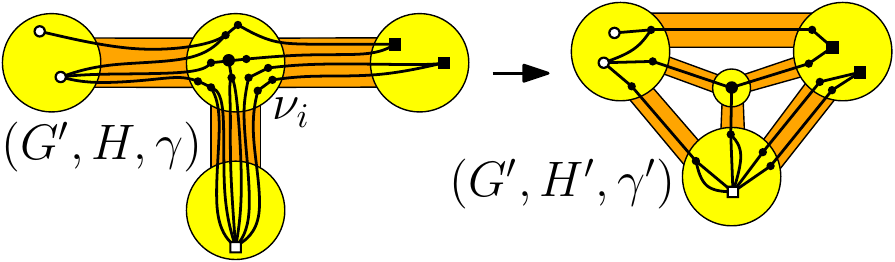}
\label{fig:derivative1}}
\caption{(a) Operation of vertex split (top) and multisplit (bottom). (b) The derivative of $(G,H,\gamma)$ in the normal form from Fig.~\ref{fig:normal}. On the left the input after splitting vertices, and on the right the obtained derivative; in the example
we have $H'=H_{\nu_i}'$, since every other $H_{\nu_j}$, for $i\not=j$, is a trivial graph with one vertex.}
\end{figure}

First, we construct the graph $G'$ by applying the following procedure
to every vertex $v\in V(G)$ such that the star $\gamma(St(v))$ has at least two edges, and thus, $v$ is not a ``spur''.
By slightly abusing the notation we will extended  $\gamma$
to take values on the vertices of $G'$.
The input $(G,H,\gamma)$ is clearly negative, if there exists a vertex $v$ in $G$ with four incident edges $vv_1,\ldots,vv_4\in E(G)$ such that $vv_1,vv_2,vv_3$ and $vv_4$ appear in the rotation at $v$ in the given order and $\gamma(vv_1)=\gamma(vv_3)\not=
\gamma(vv_2),\gamma(vv_4)$.
Otherwise, the following operations of vertex split and multisplit are applicable to $G$.

If $|E(\gamma(St(v)))|=2$, we apply the operation of vertex split to $v$ thereby turning it into an edge $uw$ as follows.
Let $E(\gamma(St(v)))=\{\rho_1,\rho_2\}$.
Let $v_1,\ldots, v_{deg(v)}$ be the neighbors of $v$.
Let $\{v_1\ldots v_l\} \cup \{v_{l+1}\ldots v_{deg(v)}\} $ be the partition of the neighbors
of $v$ such that $\gamma(vv_1)=\ldots =\gamma(vv_l)=\rho_1$ and $\gamma(vv_{l+1})=\ldots =\gamma(vv_{deg(v)})=\rho_2$.
We put $\gamma(u),\gamma(w):=\gamma(v)$, and join $u$ by an edge with the vertices in
$\{v_1\ldots v_l\}$ and $w$ with the vertices in $\{v_{l+1}\ldots v_{deg(v)}\} $.
Let $E_2\subset E(G')$ denote the set of edges in $G'$  consisting
of every edge $uw$ obtained by splitting $v\in V(G)$ such that $|E(\gamma(St(v)))|=2$.

If $|E(\gamma(St(v)))|\ge 3$, we analogously apply the operation of vertex multisplit to $v$ so that we replace $v$ with a star  $(\{v,v_1,\ldots,v_l\},\{vv_1,\ldots,vv_l\})$ with $|E(\gamma(St(v)))|=l$ edges, in which every leaf vertex $v_i$ is incident to the edges mapped by  $\gamma$ to the same edge of $H$ and $\gamma(v_i)=\gamma(v)$.
Let $V_{\ge 3}\subset V(G)$ denote the set of vertices in $G$  consisting
of the vertices $v\in V(G)$ such that $|E(\gamma(St(v)))|\ge 3$.
Note that $V_{\ge 3}$ can be treated also as a subset of $V(G')$.
Let $\mathcal{C}$ denote the set of connected components of $G'\setminus E_2\setminus V_{\ge 3}$.

Second, we construct $H'$: $V(H')= \{\rho^*| \rho\in E(H)\} \cup \{\nu_v| \ v\in V_{\ge 3} \}$, and $E(H')= \{\nu_v \rho^*|\ \rho\in E(\gamma(St(v))) \} \cup \{\gamma(C)\gamma(D)| \ 
C, D\in\mathcal{C} \  s.t. \ \mathrm{there \ exists}  \ e\in E_2  \ \mathrm{joining} \ C \ \mathrm{with} \ D\}$. We put $\gamma'(v):=\gamma(C)^*$, for  $v\in V(C)$ where $C\in\mathcal{C}$;
and $\gamma(v):=\nu_v$, for $v\in V_{\ge 3}$. Note that  $\nu_v$ and $v$ are fixed in the latter.

Finally, the embedding of $H'$, if it exists, is constructed as follows.
By F\'ary--Wagner theorem it is enough to give any embedding of $H'$ 
in the plane in a desired isotopy class, which we describe next by constructing a particular embedding of $H'$.
Let $H_{\nu_i}'$, $\nu_i\in V(H)$, denote the subgraph of $H'$ induced by $\{\rho^*| \ \rho=\nu_i\nu\in E(H) \} \cup\{\nu_v| \gamma(v)=\nu_i\}$. Let $\hat{H}_{\nu_i}'$ be obtained from $H_{\nu_i}'$ by adding 
to  $H_{\nu_i}'$ (1) the missing edges of the cycle  traversing $\{\rho^*| \ \rho=\nu_i\nu\in E(H) \}$ according to the rotation of $\nu_i$, let us denote the cycle by $C_{\nu_i}$; and (2) a new vertex joined  by the edges exactly with all the vertices of $C_{\nu_i}$.
Note that $\hat{H}_{\nu_i}'$ is vertex three-connected, and hence, if $\hat{H}_{\nu_i}'$ is planar, then the rotations at vertices in its embedding are determined up to the choice of orientation.
Note that the construction of $(G',H',\gamma')$ can be carried out in $O\left(\sum_{v\in V(G')}deg(v)\right)=O(|V(G)|)$.

Suppose that every $\hat{H}_{\nu_i}'$, for $\nu_i\in V(H)$, is a planar graph.
Let us fix for every $\nu_i\in V(H)$ an embedding of  $H_{\nu_i}'$, in which the cycle $C_{\nu_i}$  bounds the outer face and its orientation corresponds
to the rotation of $\nu_i$.
Such an embedding is obtained as a restriction of an embedding of $\hat{H}_{\nu_i}'$.
Note that for every $i$ the graph $H_{\nu_i}'$ does not have multiple edges.
Since $H$ also does not have multiple edges, $H_{\nu_i}'$ and $H_{\nu_j}'$, for $i\not=j$, are either disjoint (if $\nu_i\nu_j\not\in E(H)$) or intersect in a single vertex $(\nu_i\nu_j)^*$ (if $\nu_i\nu_j\in E(H)$). It follows that $H'$  does not have multiple edges.
The  desired  embedding of $H'$ is obtained by combining   embeddings of  $H_{\nu_i}'$, for $\nu_i\in V(H)$, in the same isotopy class as the embeddings of $H_{\nu_i}'$, that we  fixed above,
by identifying the corresponding vertices so that the restriction of the obtained embedding of $H'$ to every $H_{\nu_i}'$
has the rest of $H'$ in the interior of the outer face (of this restriction). 

\bigskip

\begin{figure}
\centering
\includegraphics[scale=1]{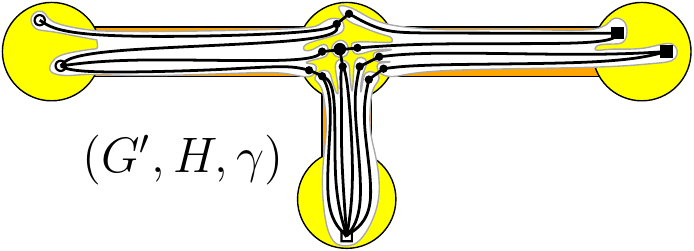}
\caption{Discs $D_\rho$ filled  white.}
\label{fig:derivative3}
\end{figure}

\begin{lemma}
\label{lemma:derivative}
The input $(G,H,\gamma)$ is  negative if one of the 
following three conditions is satisfied. 
There exists a vertex $v$ in $G$ with four incident edges $vv_1,\ldots,vv_4\in E(G)$ such that $vv_1,vv_2,vv_3$ and $vv_4$ appear in the rotation at $v$ in the given order and $\gamma(vv_1)=\gamma(vv_3)\not=
\gamma(vv_2),\gamma(vv_4)$.
The graph $\hat{H}_{\nu_i}'$, for some ${\nu_i}\in V(H)$, is not planar.
The rotation of a vertex $\nu_v\in V({H}_{\nu_i}')$, for some ${\nu_i}\in V(H)$ and $v\in V(G')$, in the obtained embedding of ${H}_{\nu_i}'$  is not  consistent with the rotation of $v$ in $G'$.

The input $(G,H,\gamma)$ is positive if and only if the input $(G',H',\gamma')$ is positive.
\end{lemma}

\begin{proof}
The first part of the claim is obvious.
For the second part, we start with ``only if'' direction, which
is easier.

To this end given an $H$-compatible embedding of $G$, we first easily construct an $H$-compatible embedding of $G'$ with respect to the input $(G',H,\gamma)$.
In the second step, for every $\rho\in E(H)$, we construct a disc $D_\rho$  containing  the restriction to $G_\rho'=\bigcup_{C\in \mathcal{C}, \gamma(C)=\rho} C$ of the $H$-compatible embedding of $G$ in its interior as follows.
Let $(G_\rho')_\times$ be a plane graph obtained from the embedding of $G_\rho'$ by turning the crossings of edges of $G_\rho'$ with both valves of $\rho$ into vertices; and parts of the valves joining closest pairs of crossing (which were turned into vertices) into edges.
The disc $D_\rho$, see Fig.~\ref{fig:derivative3}, is a small neighborhood of the union of the inner faces in the embedding of  $(G_\rho')_\times$.
Finally, we apply a homeomorphism of the plane that maps
 the union of the discs $D_\rho$'s with $G_\rho'$
into the thickening of $H'$, so that every $D_\rho$ is mapped
onto $N_\varepsilon(\rho^*)$ and a small neighborhood of every $v\in V_{\ge 3}$ onto $N_\varepsilon(\nu_v)$.  This concludes the proof of
the ``only if'' direction.

\begin{figure}
\centering
\includegraphics[scale=1]{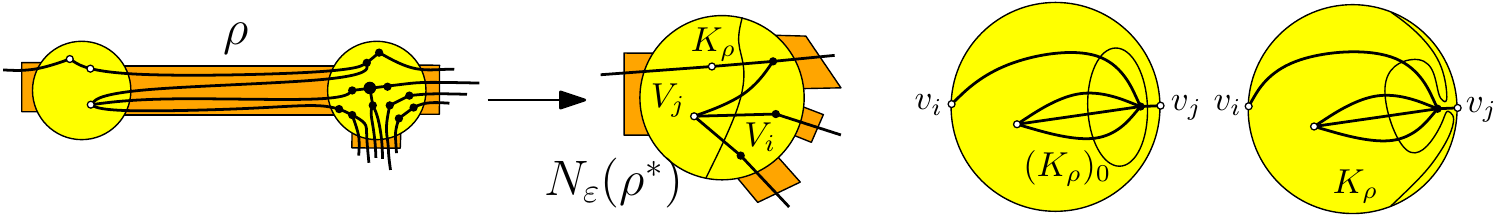}
\caption{The curve $K_\rho$ splitting the cluster of $\rho$ in the derivative on the left.
The construction of the curve  $(K_\rho)_0$ and its deformation into $K_\rho$ on the right.}
\label{fig:splittingCluster}
\end{figure}

\begin{figure}
\centering
\includegraphics[scale=1]{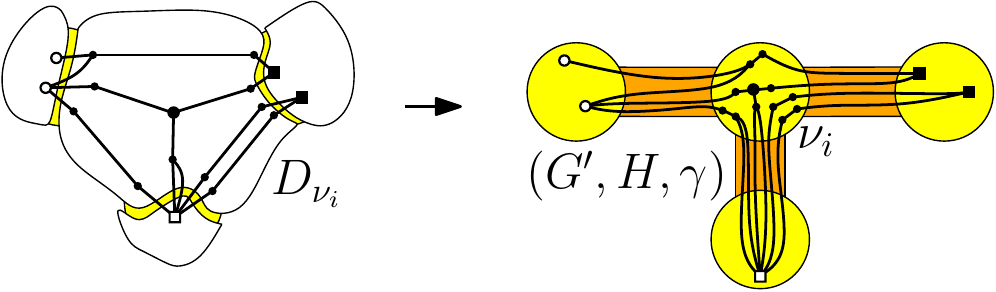}
\caption{Construction of an $H$-compatible embedding of $G'$ from  the $H'$-compatible embedding of $G'$ in Fig.~\ref{fig:derivative1}.}
\label{fig:derivative4}
\end{figure}

It remains to prove the ``if'' direction.
We show that by~\cite[Lemma 6]{FKMP15}
given an $H'$-embedding of $G'$ in the given isotopy class, every $N_{\varepsilon}(\rho^*)$, for  $\rho^*\in V(H')$, $\rho=\nu_i\nu_j$,  can be split by a simple continuous curve $K_\rho$,
 see Fig.~\ref{fig:splittingCluster}, disjoint from every edge mapped by $\gamma'$ to an edge of $H'$ into two parts as follows. The vertices in $\gamma^{-1}(\rho)$, for which $\gamma(v)=\nu_i$, are in one part and  the vertices, for which $\gamma(v)=\nu_j$, are in the other part. 

For a while suppose that  $K_\rho$'s exist. Then it follows that an $H$-compatible embedding of $G$ in the given isotopy class exists.
Analogously, to the previous paragraph, for every $\nu_i\in V(H)$, we construct a disc $D_{\nu_i}$, see Fig.~\ref{fig:derivative4},  containing the subgraph of $G'$ induced by the vertex set $\gamma^{-1}(\nu_i)$.
We construct  $D_{\nu_i}$ so that (1)
the intersection of the boundary of $D_{\nu_i}$ with the thickening of $H'$  is $\bigcup_{\rho=\nu_i\nu\in E(H)}K_\rho$, which is intersected by the boundary in the order given by the rotation at $\nu_i$;  (2)
$N_{\varepsilon}(\nu)\cap D_{\nu_i}=\emptyset$, for $\nu\not\in V(H_{\nu_i}')$; and (3) every pair of discs $D_{\nu_i}$ and
$D_{\nu_j}$, for $i\not=j$, is internally disjoint. We perturb  discs $D_{\nu_i}$'s a little bit in order to make them pairwise disjoint.
Then we apply a homeomorphism of the plane that maps  the union of the discs $D_{\nu_i}$'s with $G'$
into the thickening of $H$, so that every $D_{\nu_i}$ is mapped
onto $N_\varepsilon(\nu_i)$.
Finally, we contract the edges incident to the vertices in $V_{\ge 3}$
and contract edges in $E_2$ in order to obtain a desired
$H$-compatible embedding of $G$. 
It remains to show that $K_\rho$'s exist, which is rather simple,
but a detailed argument requires some work.
The claim essentially follows due to the fact that given an embedded
bipartite graph in the plane there exists a simple closed curve that crosses every edge of the graph exactly once.

\begin{figure}
\centering
\includegraphics[scale=1]{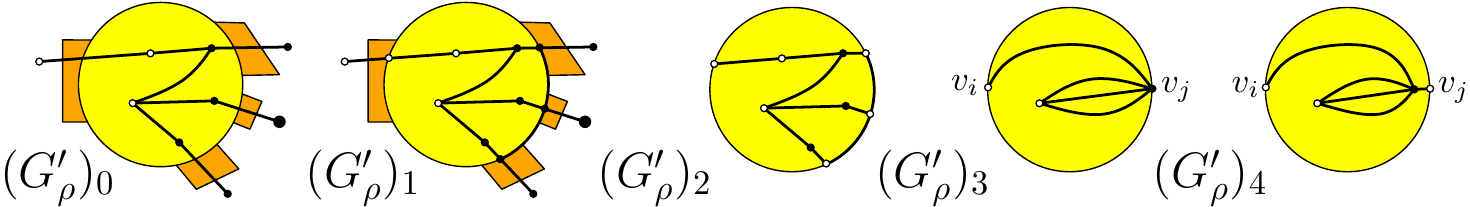}
\caption{Construction of $(G_\rho')_4$ from $G_\rho'$ in Fig.~\ref{fig:splittingCluster}.}
\label{fig:gConst}
\end{figure}

Let $G_{\rho}'$ be as above. We construct an auxiliary graph
 $(G_{\rho}')_4$ in five steps, see Fig.~\ref{fig:gConst}.
The graph $G_{\rho}'$ has the bipartition $V_i\cup V_j=V(G_\rho')$ such that $\gamma(V_i)=\nu_i$ and $\gamma(V_j)=\nu_j$, where $\rho=\nu_i\nu_j$.
Let $(G_{\rho}')_0$ be the union of  $G_{\rho}'$  with its 
incident edges in $G'$.
Let $(G_{\rho}')_1$ be a plane graph obtained from the embedding of $(G_{\rho}')_0$ by turning the crossings of edges of $G_{\rho}'$ with valves into vertices; and parts of the boundary of $N_\varepsilon(\rho^*)$ joining consecutive pairs of crossings into edges as follows.  A consecutive pair of vertices both of which are joined by an edge with a vertex of $V_i$ (or $V_j$), is joined by 
an edge contained in the boundary of $N_\varepsilon(\rho^*)$
so that the edge is disjoint from every valve of an edge in $H_{\nu_j}'$ (or $H_{\nu_i}'$).
Let $(G_{\rho}')_2$ be the subgraph of $(G_{\rho}')_1$  contained
in $N_\varepsilon(\rho^*)$.
Let $(G_{\rho}')_3$ be the plane graph obtained from $(G_{\rho}')_2$ by contracting all the edges that do not join a vertex of $V_i$ with a vertex of $V_j$.
Let $v_i$ and $v_j$ denote the vertices that resulted from the  
contractions  in the construction of $(G_{\rho}')_3$.
We assume that $v_i$ is joined by an edge with vertices in $V_j$
and $v_j$ with vertices in $V_i$. Note that $v_i$ or $v_j$ might not exist. If none of $v_i$ and $v_j$ exist we simply have 
$(G_{\rho}')_3=G_{\rho}'$.
Finally, let $(G_{\rho}')_4$ be the plane graph obtained from 
$(G_{\rho}')_3$ by applying the vertex split to $v_j$, if exists, so that the newly created edge is incident to a vertex of degree one denoted by $v_j$.
We can assume that $(G_{\rho}')_4$ is drawn in 
$N_\varepsilon(\rho^*)$ such that $v_i$ is contained in the valve
of an edge of $H_{\nu_j}'$ and $v_j$ in the valve of an edge of $H_{\nu_i}'$.

By taking the bipartition $V_i\cup\{v_i\}$ and $V_i\cup\{v_j\}$ of
$(G_{\rho}')_4$,  it follows by~\cite[Lemma 6]{FKMP15} that there exists a simple curve closed curve $(K_\rho)_0\subset N_\varepsilon(\rho^*)$ intersecting every edge of $(G_{\rho}')_4$ exactly once.
Finally, we construct $K_\rho$ by cutting and deforming $(K_\rho)_0$ as follows,
 see Fig.~\ref{fig:splittingCluster} right.  We distinguish two cases depending on whether $v_j$ exists.

First, suppose that $v_j$ exists.
The desired curve $K_\rho$ is obtained by cutting 
$(K_\rho)_0$ at its crossing point with the edge incident to $v_j$,
and applying a homeomorphism of $N_\varepsilon(\rho^*)$
that takes the severed end points very close to a pair of the  boundary points of $N_\varepsilon(\rho^*)$ that split the boundary into two parts, one of which contains the  valves of the edges in $H_{\nu_i}'$ and the other the valves of the edges in $H_{\nu_j}'$.
Second, if $v_j$ does not exist, we cut  $(K_\rho)_0$ at its
arbitrary point in the outer face of $(G_{\rho}')_4$, and apply a similar homeomorphism of $N_\varepsilon(\rho^*)$.

In the end, we extend $K_\rho$ a little bit so that both of its end points are contained in the boundary of $N_\varepsilon(\rho^*)$ and
split the contracted vertices in $(G_{\rho}')_4$ thereby recovering 
$G_{\rho}'$.
\end{proof}

\subsection{Locally injective inputs}
\label{sec:loc_inj}

The following lemma implies that by iterating the derivative at most $|E(G)|=O(|V(G)|)$ many times we obtain an input that is locally injective.

\begin{lemma}
\label{lemma:running} 
If $(G,H,\gamma)$ is in the normal form then 
 $p((G')^N,(H')^N,(\gamma')^N) \le p(G,H,\gamma)$. 
 If additionally  $(G,H,\gamma)$ is not locally 
 injective then the inequality is strict, i.e.,
 $p((G')^N,(H')^N,(\gamma')^N) <p(G,H,\gamma)$.
Moreover,  $p((G')^N,(H')^N,(\gamma')^N)\le p(G,H,\gamma)-\frac{1}{2}t$, where $t$ is the number of vertices in $G$ that do not satisfy the condition in property~(\ref{it:li}) or~(\ref{it:v32}) of locally injective inputs.
\end{lemma}
\begin{proof}
Note that edges incident to fixed vertices in $((G')^N,(H')^N,(\gamma')^N)$ do not contribute towards $p((G')^N,(H')^N,(\gamma')^N)$, and thus, we will deal only with the remaining edges.
We consider vertices in $V_{\ge 3}\subseteq V(G')$ to be their corresponding  vertices in $V((G')^N)$.
Since suppressing the vertices of degree two in $G'$ and $H'$ violating property~(\ref{it:supp}) of the normal form in order to make the property satisfied does not increase the value of the potential, for the purpose of the proof of the lemma by somewhat abusing the notation we assume that we keep such vertices in $(G')^N$ and $(H')^N$.

 Let $H_0'$ be the subgraph of $(H')^N=H'$ induced by its vertex subset $\{\rho^*| \ \rho\in E(H)\}$. Every connected graph on $n$ vertices  has at least $n-1$ edges. It follows that ($\diamond$) the number of edges in $H_0'$ is at least $|V(H_0')|-c= |E(H)|-c$, where $c$ is the number of connected components of $H_0'$ that are trees. We use this fact together with the following observation to prove the claim.
 
 Suppose for a while that $H_0'$ is connected.
 The set of edges of $(G')^N$ not incident to any $v\in V_{\ge 3}$, where the vertices in $V_{\ge 3}$ are now fixed, forms a matching $M'$ whose edges are in one-to-one correspondence with edges in $E_2$ in $G'$.
Note that none of the end vertices of edges in $E_2$ is of degree one in $G'$. 
Let $I_2=\{(v,e_2)\ |v\in e_2\in E_2\}$.
By using the natural one-to-one correspondence between the edges of $G'\setminus E_2\setminus V_{\ge 3}$ and the edges of $G$,
it follows that the size of $M'$ is upper bounded by the size of $E(G)$, since
 $2|M'|=2|E_2|=|I_2|\le |\{(v,e)| \ e\cap e_2=\{v\}, e_2\in E_2, e\in E(G') \}| \le 2|E(G)|$.
 Hence, it follows that 
 \begin{equation}
 \label{eqn:11}
|M'|=|E_2|\leq |E(G)|
\end{equation} 
 
Furthermore,  $|M'|=|E(G)|$, only if $(G,H,\gamma)$ is locally injective, and $H_0'$ contains a cycle.
Indeed,  if $H_0'$ does not contain a cycle, it
is either a trivial graph consisting of a single vertex, or it contains a pair of vertices $(\rho_0)^*$ and $(\rho_1)^*$ of degree one such that 
$\gamma'(v_0)=(\rho_0)^*$ and 
$\gamma'(v_1)=(\rho_1)^*$, where $v_0\in e_0\in E_2, v_0\in f_0=v_0u_0\not\in E_2$, and $v_1\in e_1\in E_2, v_1\in f_1=v_1u_1\not\in E_2$. It holds that
$|M'|<|E(G)|$, because $(u_0,f_0)$ and  $(u_1,f_1)$ 
are not in the image of the injective map $\mu$
from $I_2$ taking $(v,e_2)$, $v\in e_2\in E_2$ to a pair $(v,e)$, $e\in E(G')$, such that $e\cap e_2=\{v\}$. Note that there exists at least two such pairs also if $H_0'$ is trivial (which is a fact that we will need later).  Namely, $(u,uv)$ and $(v,uv)$, for some $\gamma'(uv)= \rho^*\in V(H_0')$. By the same token, we have that $|M'|<|E(G)|$, if $(G,H,\gamma)$ is not locally injective. In fact, if $|M'|=|E(G)|$ then every connected component of $G'$ must be a cycle.

If $H_0'$ has more connected components, we then  have $|M'|\le |E(G)|-c$, where the inequality is strict  if $(G,H,\gamma)$ is not locally injective.
Indeed, if $|M'|= |E(G)|-c$, then there exist
exactly $2c$ pairs $(v,e)$, $v\in e\in E(G')\setminus E_2\setminus V_{\ge 3}$, that are not 
in the image of the map $\mu$.
However, we showed in the previous paragraph that 
there are at least $2c$ such pairs 
$(u,f)$, where both $u$ and $f$ are mapped by $\gamma'$ to a vertex of degree at most one in $H_0'$. Hence, if $|M'|= |E(G)|-c$ then all the pairs, that are not contained in the image of $\mu$, are accounted for by such $(u,f)$'s, in which case 
$V_{\ge 3}$ is exactly the subset of $V(G')$ of vertices of degree at least three. This 
establishes property~(\ref{it:li}) of locally injective inputs.
Finally, to establish also~(\ref{it:v32}) we consider the natural correspondence  of the vertices of degree one in $G$ with the vertices of degree one in $G'$.
Note that none of the pairs $(u,f)$, $u\in f\in E(G')$, where $u$ is of degree one, is in the image of $\nu$. 
Thus, if $|M'|= |E(G)|-c$, then every leaf $u\in V(G')$ is  mapped by $\gamma'$ to some $\rho^*$, $\rho\in E(H)$, which is an isolated vertex 
or a leaf of $H_0'$. We need  to show that 
there is no other edge besides $f\ni u$ mapped by $\gamma' $ to $\rho^*$. If $\rho^*$ is an isolated vertex of $H_0'$ this is immediate, since otherwise $|M'|< |E(G)|-c$. If
$\rho^*$ is a leaf of $H_0'$, every other edge $g\not=f$ such that $\gamma'(g)=\rho^*$ must share both end vertices with an edge of $E_2$,
but then $\rho^*$ has degree at least two  in $H_0'$ (contradiction).

Putting it together,  we have $|M'|\le |E(G)|-c$ and 
($\diamond$) $|E(H)|-c \le |E(H_0')|$, where the first inequality is strict if $(G,H,\gamma)$ is not locally injective as we just showed. 
Since the remaining edges of $(G')^N$ and $(H')^N$ contributes together zero towards $p((G')^N,(H')^N,(\gamma')^N)$, summing up the inequalities concludes the proof.

The ``moreover'' part follows immediately due to the  fact that every vertex of $G$ not satisfying~(\ref{it:li}) or~(\ref{it:v32}) causes the slack of $\frac 12$ in~(\ref{eqn:11})
as shown by the previous analysis.
\end{proof}

Given an input $(G,H,\gamma)$ in the normal form.
Similarly as in Section~\ref{sec:derivative}, let $V_{\ge 3}\subseteq V(G)$ denote the set of vertices in $G$  consisting of the vertices $v\in V(G)$ such that $|E(\gamma(St(v)))|\ge3$.
The input is \emphh{strongly locally injective} if 
it is locally injective  and

\begin{enumerate}[(iii)]
\item \label{it:v3} 
every vertex in $V_{\ge 3}$ is fixed.
\end{enumerate}

For convenience, we would like work
with strongly locally injective inputs, see Fig.~\ref{fig:derivative2}.
The following lemma shows that if the input $(G,H,\gamma)$ is locally injective, but not strongly, we just derivate it one more time in order to arrive at a strongly locally injective input.

\begin{figure}
\centering
\includegraphics[scale=1]{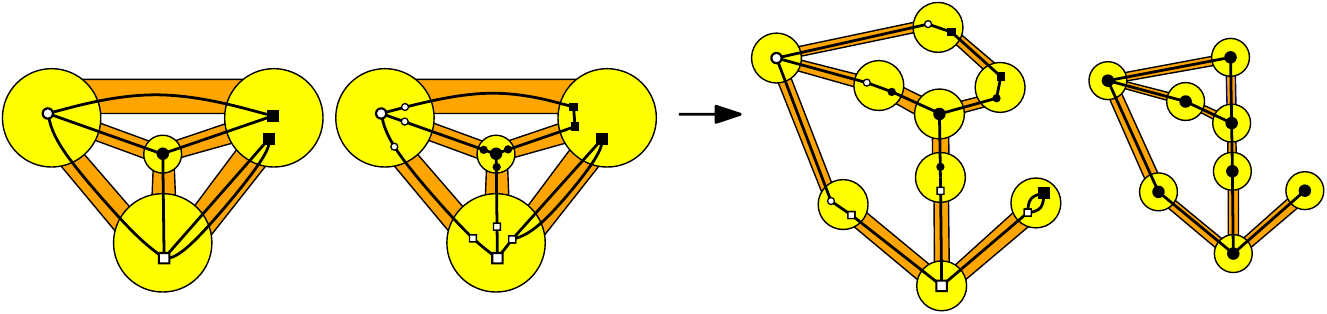}
\caption{Constructing the normal form and derivating one more time the derivative from Fig.~\ref{fig:derivative1}
on the left, we obtain an input that is strongly locally injective in the normal form on the right.}
\label{fig:derivative2}
\end{figure}

\begin{lemma}
\label{lemma:strong} 
Suppose that $(G,H,\gamma)$ in the normal form is locally injective.
Then in $((G')^N,(H')^N,(\gamma')^N)$, every vertex $v\in V((G')^N)$, such that $|E((\gamma')^N(St(v)))|\ge 3$ is fixed.
Moreover, $((G')^N,(H')^N,(\gamma')^N)$ is still locally injective.
\end{lemma}

\begin{proof}
The lemma follows directly from the definition of the derivative.
\end{proof}

Deciding in, roughly,  quadratic time in $p(G,H,\gamma)$,  which is sufficient for us
due to the bottleneck discussed in Section~\ref{sec:alg}, whether the strongly locally injective input $(G,H,\gamma)$ is  positive, is quite straightforward.
The reason is that in this case the order of  crossings of a valve with edges, that are incident to the same vertex $v$ of $G$, along the valve  in an $H$-compatible embedding of $G$ is determined by the rotation at $v$.
In order to decide if a desired $H$-compatible embedding of $G$ exists, we just detect if for every valve $\omega$ such an order of all the edges crossing $\omega$ exists, such that together the orders are compatible.
To this end we consider relations between unordered pairs of edges of $G$ such that the edges in a pair are mapped by $\gamma$ to the same edge of $H$,
and two pairs are related if they intersect in a pair of vertices.
In the following we assume that $(G,H,\gamma)$ is strongly locally injective.

 Let $\Xi=\{\{e,f\}| \ e,f\in E(G) \ s.t. \ e\not= f \ \mathrm{and} \  \gamma(e)=\gamma(f) \}$.
Two elements $\{e_1,f_1\}\in \Xi$ and $\{e_2,f_2\}\in \Xi$ are \emphh{neighboring} if
$|e_1\cap e_2|=1$, $|f_1\cap f_2|=1$ and $\gamma(e_1 \cap e_2)= \gamma(f_1 \cap f_2)$; we write 
 $\{e_1,f_1\} \sim \{e_2,f_2\}$.
An element $\{e_1,f_1\}\in \Xi$ is a \emphh{boundary pair} if there exists at most
one $\{e_2,f_2\}\in \Xi$  such that  $\{e_1,f_1\}$ and $\{e_2,f_2\}$ are neighboring.
Let $\Xi_1,\ldots \Xi_l$ be equivalence classes of the transitive closure of the relation $\sim$.
A boundary pair $\{e_1,f_1\}\in \Xi$ is \emphh{determined} if there exists
a pair of edges $e_2$ and $f_2$ such that $|e_1\cap e_2|=1$, $|f_1\cap f_2|=1$, $\gamma(e_1 \cap e_2)= \gamma(f_1 \cap f_2)$ and $\gamma(e_2)\not=\gamma(f_2)$.
By properties~(\ref{it:li}) and~(iii), the subgraph $G_{\Xi}$ of $G$ induced by $\bigcup_{\{e,f\}\in \Xi}\{e,f\}$ has maximum degree two.
First, we consider the case when a connected component of $G_{\Xi}$ does not contain a vertex of degree one.

\begin{lemma}
\label{lemma:badCycle}
If there exists an equivalence class $\Xi_c$, such that the subgraph $G_{\Xi_c}$ of $G$ induced by $\bigcup_{\{e,f\}\in \Xi_c}\{e,f\}$ is a cycle, then $(G,H,\gamma)$ is a negative input.
\end{lemma}

\begin{proof}
Let $\Xi_c=\{\{e_0,f_0\},\ldots \{e_{m-1},f_{m-1}\}\}$, where
$\{e_p,f_p\} \sim \{e_{p+1 \mod m},f_{p+1 \mod m}\}$.
Since $\bigcup_{\{e,f\}\in \Xi_c}\{e,f\}$ induces a cycle $C$ of $G$, there exists the minimum value $p_0<m$ such
that $e_{p_0}=f_0$ or $f_{p_0}=e_0$.
Note that $p_0|m$, since $\gamma(e_{q})=\gamma(e_{q+ap_0 \mod m})$, for every $0\leq q, a<m$, and 
that $C=e_0,\ldots, e_{m-1}$.
Due to the fact that the plane is orientable, it follows that the cycle $C$ does not admit an $H$-compatible emebedding, since in an $H$-compatible embedding $C$ must wind around a point in the plane more than once.
\end{proof}

Note that Lemma~\ref{lemma:badCycle} does not cover only the case when $G_{\Xi_c}$ is a union of two cycles.
 By~(\ref{it:v32}), it must be that if $\Xi_c$ contains a boundary pair, then it, in fact,
contains exactly two boundary pairs, both of which are determined.
Hence, in the following we assume that every $\Xi_c$ either gives rise to a pair of cycles, or contains exactly two determined boundary pairs.
We construct for every valve $\omega$ of $\rho\in E(H)$ the relation
$(E_{\rho}, <_{\omega})$, where
$E_{\rho}=\{e\in E(H)|  \ \gamma(e)=\rho\}$.
We define relations $(E_{\rho}, <_{\omega})$ by propagating  relations enforced by the determined boundary pairs, for every  determined pair contained in $\Xi$.
 We assume that 
 $(E_{\rho}, <_{\omega})$ encodes the increasing order of the crossing points 
 of edges with $\omega$ as encountered when traversing  $\omega\subset N_\varepsilon(\nu)$ in
 the direction inherited from the counterclockwise orientation of the boundary of $N_\varepsilon(\nu)$.

Let $\{e_1,f_1\}\in \Xi_c\subseteq \Xi$ be {determined}.
Let $\Xi_c=\{\{e_1,f_1\},\ldots \{e_m,f_m\}\}$ such that  $\{e_p,f_p\} \sim \{e_{p+1},f_{p+1}\}$.
Let $\gamma(e_1)=\gamma(f_1)=\nu_i\nu_j,\gamma(e_0)=\nu_i\nu_{j'}, \gamma(f_0)=\nu_i\nu_{j''}$, where $\nu_{j'}\not=\nu_{j''}$
and $|e_0\cap e_1|=1$ and $|f_0\cap f_1|=1$.
W.l.o.g. we suppose that $\nu_i\nu_j, \nu_i\nu_{j'}$ and $\nu_i\nu_{j''}$ appear in the rotation of $\nu_i$ in this order counterclockwise.
Let $\omega_1$ be the valve of $\nu_i\nu_j$ at $\nu_i$. 
Let $\omega_2$ be the valve of $\nu_i\nu_j$ at $\nu_j$. 
We put the  relation $f_1<_{\omega_1}e_1$ into  $(E_{\nu_i\nu_j}, <_{\omega_1})$ and $e_1<_{\omega_2}f_1$  into $(E_{\nu_i\nu_j}, <_{\omega_2})$.
Recursively, we put $f_{p+1}<_{\omega_{2p+1}}e_{p+1}$ into  $(E_{\nu_i\nu_j}, <_{\omega_{2p+1}})$ and $e_{p+1}<_{\omega_{2(p+1)}}f_{p+1}$  into $(E_{\nu_i\nu_j}, <_{\omega_{2(p+1)}})$,
if $f_{p}<_{\omega_{2p-1}}e_{p}$ and $e_{p}<_{\omega_{2p}}f_{p}$, and vice-versa,
where $\omega_{2p}$ and $\omega_{2p+1}$ are valves contained in the boundary of the same
disc.

If $G_{\Xi_c}$ is
a union of two disjoint cycles we add $f_{p}<_{\omega_{2p}}e_{p}$
and $e_{p}<_{\omega_{2p-1}}f_{p}$, or
 $f_{p}>_{\omega_{2p}}e_{p}$
and $e_{p}>_{\omega_{2p-1}}f_{p}$ for every $p$, in correspondence with the isotopy class of $G$.

\begin{lemma}
\label{lemma:relations}
Suppose that every equivalence class $\Xi_c$ contains exactly two determined boundary pairs or $G_{\Xi_c}$ is
a union of two disjoint cycles.
 We can test in $O((p(G,H,\gamma))^2+|V(G)|)$ time if $(G,H,\gamma)$ is positive or negative.
\end{lemma}

\begin{proof}
The relations $(E_{\rho}, <_{\omega})$ can be clearly constructed in $O((p(G,H,\gamma)^2)$ time,
since only edges not incident to fixed vertices are contained in pairs of $\Xi$.
If the constructed $(E_{\rho}, <_{\omega})$ is a total order for all $\rho\in E(H)$
and its valve $\omega$, the isotopy class of every  $H$-compatible embedding of $G$ is determined by an embedding constructed as follows.  We first draw the crossings of valves with edges of $G$ according to the orders $(E_{\rho}, <_{\omega})$; join every pair of consecutive crossing
on the same edge of $G$ by a straight-line segment contained in a pipe of an edge of $H$; and finish by drawing the straight-line segments joining vertices  of $G$ with the already drawn parts of edges contained in pipes.
It is enough to check if the obtained embedding
is in the desired isotopy class of $G$, which can be easily done in 
 $O(|V(G)|)$ time by traversing orders $(E_{\rho}, <_{\omega})$.
 Note that the only thing that can make the input negative is the containment of connected components of $G$ in the interiors of its faces.

If  $(E_{\rho}, <_{\omega})$, for some $\rho\in E(H)$, contains a cyclic chain of inequalities, the input is clearly negative.
\end{proof}

\subsection{Algorithm}
\label{sec:alg}

We start with a description of the decision algorithm proving the first part of the theorem. \\

{\bf Decision Algorithm.}
Let $(G,H,\gamma)=(G_0,H_0,\gamma_0)$ be the input.
We work with inputs in which $G$ contains multiple edges and loops. However, w.l.o.g we assume that $G$ does not contain a pair of multiple edges joining the same pair of vertices 
that form a face of size two, i.e., a lens,
except for the outer face. 
Moreover, we assume that whenever a lens is created during the execution of the algorithm, 
the lens is eliminated by deleting one of its edges. 

An execution of the algorithm is divided into steps.
During the $s$-th step we process $(G_s,H_s,\gamma_s)$
and output $(G_{s+1},H_{s+1},\gamma_{s+1})$ as follows.

First, by following the procedure described in Section~\ref{sec:instances} we either construct an instance $((G_s)^N,(H_s)^N,(\gamma_s)^N)$ in the normal form that is positive if and only if $(G_s,H_s,\gamma_s)$ is positive,  or output that $(G,H,\gamma)$ is negative, if the hypothesis of the first part of Lemma~\ref{lemma:normal} is satisfied.

 Second, if  $((G_s)^N,(H_s)^N,(\gamma_s)^N)$ is not strongly locally injective  we proceed as follows.
 If $(G_s,H_s,\gamma_s)$ satisfies the hypothesis of the first part
of Lemma~\ref{lemma:derivative} with $((G_s)^N,(H_s)^N,(\gamma_s)^N)$ playing the role of $(G,H,\gamma)$ we output
 that $(G,H,\gamma)$ is negative; otherwise we  construct the derivative $((G_{s}^N)',(H_s^N)',(\gamma_s^N)')=(G_{s+1},H_{s+1},\gamma_{s+1})$ defined in Section~\ref{sec:derivative} and proceed to the $(s+1)$-st step.
Otherwise,  $((G_s)^N,(H_s)^N,(\gamma_s)^N)$ is strongly locally injective and we construct equivalence classes $\Xi_1,\ldots \Xi_l$ 
from  Section~\ref{sec:loc_inj} defined by $((G_s)^N,(H_s)^N,(\gamma_s)^N)$
and proceed as follows.

We check if there exists a class $\Xi_c$ satisfying the hypothesis of Lemma~\ref{lemma:badCycle}. If this is the case, then we output that $(G,H,\gamma)$ is negative.
 Otherwise, we construct relations $(E_{\rho}, <_{\omega})$, for every $\rho\in (H_s)^N$ and its valve  $\omega$.
If there exists $(E_{\rho}, <_{\omega})$ that is not a total order
we output that $(G,H,\gamma)$ is negative;
otherwise we check if the isotopy class of an $H$-compatible embedding of $G_s$ enforced by 
 relations $(E_{\rho}, <_{\omega})$ is the same as the given one
 and output that  $(G,H,\gamma)$ is positive if and only if this is the case.

The correctness of the algorithm follows directly from Lemma~\ref{lemma:normal},\ref{lemma:derivative},\ref{lemma:badCycle}, and~\ref{lemma:relations}. \\

{\bf Running time analysis.}
By Lemma~\ref{lemma:running}, Lemma~\ref{lemma:strong}, and Lemma~\ref{lemma:linear}
the number of steps of our algorithm is $O(|V(G)|)$.
Furthermore, we show that $|V((G_{s})^N)|=O(|V(G)|)$,
for every $s$.

By Lemma~\ref{lemma:running}, $p((G_{s+1})^N,(H_{s+1})^N,(\gamma_{s+1})^N)\le p((G_s)^N,(H_s)^N,(\gamma_s)^N)-\frac{1}{2}t$, where $t$ is the number of vertices in $(G_s)^N$ that do not satisfy the condition in property~(\ref{it:li})  or~(\ref{it:v32}) of locally injective inputs. 
The number of newly created vertices in $(G_{s+1})^N$ of degree at least three satisfying the condition in property~(\ref{it:li}) during the $s$-th  step of the algorithm  is at most $t=
2(p((G_s)^N,(H_s)^N,(\gamma_s)^N)-p((G_{s+1})^N,(H_{s+1})^N,(\gamma_{s+1})^N))$.
Note that a vertex of degree $d\ge 3$ in 
$((G_{s})^N,(H_{s})^N,(\gamma_{s})^N)$ satisfying the condition of property~(\ref{it:li}) becomes a fixed vertex of degree $d$ in $((G_{s+1})^N,(H_{s+1})^N,(\gamma_{s+1})^N)$.

Let $V_{\ge 3}$ denote the set of fixed vertices of degree at least three in $(G_s)^N$.
By the previous paragraph, $|V_{\ge 3}|\le |V(G)|+ \sum_{s\ge 0}2(p((G_s)^N,(H_s)^N,(\gamma_s)^N)-p((G_{s+1})^N,(H_{s+1})^N,(\gamma_{s+1})^N))\le |V(G)|+ 2p((G_0)^N,(H_0)^N,(\gamma_0)^N)= O(|E(G)|)=O(|V(G)|)$,
for every $s$, due to the definition of the potential.
The fact $|V((G_{s})^N)|=O(|V(G)|)$, for every $s$, then follows by~(\ref{it:supp}) in the definition of the normal form.   
Indeed, the number of  vertices of degree two in $(G_s)^N$ mapped to redundant vertices in $(H_s)^N$ is linear in the number of remaining vertices in $(G_s)^N$ due to Lemma~\ref{lemma:linear}.
Hence, the number of vertices in $(G_{s})^N$ that are not fixed vertices of degree at least three is linear  in $p(G,H,\gamma)+|V((H_s)^N)|=p(G,H,\gamma)+|V(H_0')|+|V_{\ge 3}|$ due to Lemma~\ref{lemma:linear}, where $H_0'$ is defined as in the proof of Lemma~\ref{lemma:running} with $((G_{s-1})^N,(H_{s-1})^N,(\gamma_{s-1})^N)$ playing the role of $(G,H,\gamma)$.
It follows that the number of vertices in $|V((G_s)^N)|$, for every $s$, is linear in $p(G,H,\gamma)+|V(G)|$, since the size of the subset of $V(H_0')$ of non-redundant vertices is upper bounded by $|E(G)|=O(|V(G)|)$.
This more-or-less follows  inductively from~(\ref{eqn:11}) in the proof of Lemma~\ref{lemma:running},
except that in every step we consider  the subgraph of $(G_{s'})^N$, $s'< s$, induced by the set $E_2$, where $E_2$ is defined with $((G_{s'})^N,(H_{s'})^N,(\gamma_{s'})^N)$ playing the role of $(G,H,\gamma)$. Indeed, the rest of the edges in $G_{s'+1}$ are incident to fixed vertices and are mapped by $\gamma_{s'+1}$ to a redundant vertex of $H_{s'+1}$. 
Formally, we show by induction on $s'$ that  $|E((G_{s'})^N\setminus V_{s'})|\le |E(G)|$, where $V_{s'}$ is the set of fixed vertices of degree at least two in $V((G_{s'})^N)$ for  all $s'$.
In the base case we have $|E((G_{1})^N\setminus V_{1})|\le|E_2|\le |E(G)|$ by~(\ref{eqn:11}). 
For $s'>0$, we have $|E((G_{s'+1})^N\setminus V_{s'+1})|\le |E((G_{s'})^N\setminus V_{s'})|$ by the argument that
we used to prove (\ref{eqn:11}).

After having shown that $O(|V((G_{s-1})^N)|)=O(|V(G)|)$, it is easy to see that the $s$-th step of the algorithm can be easily carried out in $O(|V(G)|)$ time, since the planarity testing  and embedding construction
of all $\hat{H}_{\nu_i}'$'s needed in the construction of the derivative
can be done in linear time in $O(|V(G_s)|)=O(|V((G_{s-1})^N)|)$~\cite{HoTa74_planarity}, and the construction of the instance in the normal form from the given one takes the same running time.
The last step of the algorithm in which we construct orders $(E_{\rho}, <_{\omega})$
can be easily done in $O(|V(G)|^2)$, since the 
number of pairs in $\Xi$ is $O(|V(G)|^2)$ due to $|V(G_{s})|=O(|V(G)|)$ .  \\

{\bf Algorithm constructing an embedding.}
The construction of an  $H$-compatible embedding of $G$ for strongly locally injective inputs  is given
by the set of total orders $(E_{\rho}, <_{\omega})$, for every $\rho\in H$ and a valve $\omega$ of $\rho$. 
Therefore in order to construct a desired $H$-compatible embedding of $G$ we 
need to reverse the order of steps in the decision algorithm.
To this end we make the proof of the second part of  
Lemma~\ref{lemma:derivative} algorithmic.
In order words, we need to construct the order in which a curve $K_{\rho}$ intersect  edges 
of $G'$.
Since we can construct a desired order  in linear time
by the following lemma, the overall quadratic running time follows.

\begin{lemma}
\label{lemma:orderCyclic}
Given a plane bipartite graph $G$, we can construct in  $O(|V(G)|)$ time a cyclic order $\mathcal{O}$ of edges of $G$ such that there exists a simple closed curve in the plane properly crossing every edge of $G$ exactly once,
but otherwise disjoint from $G$, in the order given by $\mathcal{O}$.
\end{lemma}

\begin{proof}
Let $V_1\uplus V_2=V(G)$ be the bipartition of $V(G)$.
We construct a plane graph $G\cup T$ such that $V(T)=V_1$,
$E(T)\cap E(G)=\emptyset$, and $T$ is a spanning tree of $(G\cup T)[V_1]=T$.
The tree $T$ is constructed in linear time as follows. We subdivide every face of $G$ by as many edges as possible in ${V_1\choose 2}$, while keeping the resulting graph $G_0$ plane and its subgraph $G_0[V_1]$ without multiple edges, e.g., we perform the subdivision so that every subgraph of $G_0[V_1]$ subdiving a face of $G$ is a star containing all the vertices incident to the  face.
Since  rotations at vertices are stored in doubly linked lists, $G_0$ can be constructed in time $O(|V(G)|)$.
Note that $G_0[V_1]$ is connected,
since every face of $G$ is incident to a vertex in $V_1$.
 The tree $T$ is obtained  as a spanning tree of $G_0[V_1]$.

We contract all the edges of  $T$ in $G\cup T$.
Let $\mathcal{O}'$ be  the rotation at the vertex, that $T$ was contracted into, in the resulting graph.
The desired order $\mathcal{O}$ is obtained 
by substituting in $\mathcal{O}'$ for every edge its corresponding edge in $G$.
\end{proof}

\bibliographystyle{plain}


\bibliography{bib}


\end{document}